\documentclass[aps,twocolumn,superscriptaddress,showpacs,amsfonts,amsmath,amsthm,amssymb,mathtools,latexsym,nofootinbib]{revtex4-1}
\usepackage[colorlinks,linkcolor=blue,urlcolor=blue,citecolor=blue,bookmarks,bookmarksnumbered]{hyperref}
\usepackage{breqn}
\usepackage{bm,braket,dcolumn}
\usepackage{slashed}
\usepackage{graphicx,subfigure}
\usepackage{relsize}
\usepackage{epstopdf,epsfig}
\usepackage{enumitem}
\usepackage{amsthm}
\usepackage{braket}
\usepackage{tikz-cd,amscd}

\theoremstyle{definition}

\newtheorem*{claim*}{Claim}

\graphicspath{{pics/}{}}

\newcommand{\Z}{\mathbb{Z}}
\newcommand{\N}{\mathbb{N}}
\newcommand{\C}{\mathbb{C}}
\newcommand{\R}{\mathbb{R}}

\newcommand\defeq{\mathrel{\stackrel{\makebox[0pt]{\mbox{\normalfont\tiny def}}}{=}}}

\newcommand{\refeq}[1]{Eq.~\eqref{#1}}
\newcommand{\refsec}[1]{Sec.~\ref{#1}}
\newcommand{\reffig}[1]{Fig.~\ref{#1}}

\newcommand{\Conf}{\operatorname{Conf}}
\newcommand{\PB}{\operatorname{PB}}

\begin{document}
\newcommand*{\PITT}{Department of Physics and Astronomy, University of Pittsburgh, Pittsburgh, Pennsylvania 15260, United States}\affiliation{\PITT}
\newcommand*{\PQI}{Pittsburgh Quantum Institute, Pittsburgh, Pennsylvania 15260, United States}
\newcommand*{\PI}{Perimeter Institute for Theoretical Physics, Waterloo, Ontario N2L 2Y5, Canada}
\title{Topological classification of non-Hermitian bands}

\author{Zhi Li}\affiliation{\PITT}\affiliation{\PQI}\affiliation{\PI}
\author{Roger S. K. Mong}\affiliation{\PITT}\affiliation{\PQI}

\begin{abstract}
We proposed a framework for the topological classification of non-Hermitian systems. Different from previous $K$-theoretical approaches, our approach is a homotopy classification, which enables us to see more topological invariants. Specifically, we considered the classification of non-Hermitian systems with separable band structures. We found that the whole classification set is decomposed into several sectors based on the \textbf{braiding} of energy levels and characterized by some \textbf{braid group} data. Each sector can be further classified based on the topology of eigenstates (wave functions). Due to the interplay between energy levels braiding and eigenstates topology, we found some torsion invariants, which only appear in the non-Hermitian world via homotopical approach. We further proved that these new topological invariants are \textbf{unstable} (fragile), in the sense that adding more bands will trivialize these invariants.

\end{abstract}
\maketitle

\section{Introduction}

While we are used to assuming the Hermiticity of Hamiltonians, as required by the axioms of quantum mechanics, there has been growing interest in non-Hermitian Hamiltonians these years.  Indeed, in the Hermitian quantum mechanics framework, non-Hermitian Hamiltonian can emerge as an effective description of open systems with gain and loss \cite{moiseyev_2011,PhysRevLett.70.2273,Rotter_2009,Diehl2011,PhysRevB.84.205128_2011,Regensburger2012,Zhen2015,Feng2017,RevModPhys.87.61,El-Ganainy2018,PhysRevLett.100.103904,PhysRevLett.101.080402,PhysRevLett.103.093902,PhysRevLett.103.123601,PhysRevLett.104.153601,PhysRevLett.106.213901,PhysRevLett.108.024101,PhysRevLett.108.173901,PhysRevLett.115.200402,PhysRevLett.116.133903_2016,PhysRevLett.118.040401} or systems with finite-lifetime quasiparticles or non-Hermitian self energy \cite{2017arXiv170805841K,PhysRevB.98.035141,PhysRevB.97.041203}, which can be experimentally realized in atomic or optical systems \cite{PhysRevLett.86.787,Ruter2010,Feng2012,Chang2014,Hodaei975,Gao2015,Hodaei2017}. Moreover, non-Hermitian Hamiltonians with certain properties can serve as an extension of conventional Hermitian quantum mechanics \cite{PhysRevLett.80.5243,PhysRevLett.89.270401,makesenseofNH}.

Classification of topological phases of matter has been one of the central problems in condensed matter physics for the last two decades. While a complete classification of topological phases is still in progress, the classification for gapped non-interacting fermions is well-established \cite{AZ,Shinsei,KitaevK,Ryu_2010} based on the geometry and topology of the band structures.

Inspired by the great success in topological phases for Hermitian systems, there have been lots of works focusing on the topological aspects of non-Hermitian systems \cite{Gong_2018,symtop_2018,PhysRevB.99.125103,PhysRevLett.123.066405_2019,PhysRevB.99.235112,Ghatak_2019,sun2019nontrivial}. On the one hand, many familiar constructions for topological phases can be extended in the case of non-Hermiticity. For example, people have constructed the non-Hermitian counterparts for the Su-Schrieffer-Heeger model \cite{PhysRevB.84.205128_2011,PhysRevA.89.062102_2014,PhysRevLett.116.133903_2016,PhysRevB.97.045106_2018,PhysRevA.97.052115_2018,2019arXiv190604700J}, Chern insulators \cite{PhysRevLett.121.136802,PhysRevB.98.155430_2018,PhysRevB.98.245130_2018,Fu,PhysRevB.98.165148_2018}, and quantum spin Hall effects \cite{nonHQSHE}. On the other hand,  non-Hermitian systems also exhibit many unusual phenomena with no counterpart in the Hermitian world. These include exceptional points \cite{Keldysh_1971,Berry2004,Heiss_2012}, anomalous bulk-edge correspondence \cite{Xiong_2018,PhysRevLett.116.133903_2016,PhysRevB.97.121401_2018,PhysRevLett.121.026808_2018,Gil_2019}, non-Hermitian skin effect \cite{PhysRevLett.121.086803_2018,PhysRevLett.121.136802}, and sensitivity to boundary conditions \cite{Xiong_2018,PhysRevB.99.201103_2019}. For a recent review, see Ref.~\onlinecite{MartinezAlvarez2018} and references therein.

There have been some works \cite{Gong_2018,symtop_2018,PhysRevB.99.125103,PhysRevLett.123.066405_2019,PhysRevB.99.235112} on the general classification of non-Hermitian systems, aiming at a generalization of the Hermitian periodicity table \cite{Shinsei,KitaevK}. In these works, the authors first determined reasonable symmetry classes in the non-Hermitian setting (a generalization of Ref.~\onlinecite{AZ}), then used a unitarization/Hermitianization map to reduce the problem into the Hermitian setting where one can apply $K$-theory.

In this article, we proposed a more conceptually straightforward homotopical~\cite{Simon,JoelMoore,Hopfinsulator,Roy,Roger,PhysRevB.98.094432} framework towards the topological classification of non-Hermitian band structures, which enables us to see more topological invariants beyond $K$-theoretical approaches. With rigorous algebraic-topological calculation, we implemented our idea in detail for systems with no symmetry.

We found that, due to the non-Hermiticity of the Hamiltonian, energy levels can be complex and therefore braid with each other in the complex plane, which decomposes the whole classification set into several \textbf{braiding sectors}. Each sector can be further classified based on the topology of eigenstates (wave functions), akin to the usual topological classification for Hermitian systems, but with more subtleties coming from the braiding of energy levels and the shape of the Brillouin zone (torus vs. sphere). We found some new torsion invariants (for example, $\Z_2$), and a physical explanation of these new invariants is given.

We also considered the stability of these new invariants, in the sense that whether adding other bands will trivialize these invariants, even if the band has no crossing with previous bands. Similar to the $\Z$ invariants of Hopf insulators \cite{Hopfinsulator}, our torsion invariants are unstable. We managed to give a combinatorial proof for instability in general. The physical origin of the instability is also discussed.

This article is organized as follows. In \refsec{sec-setting}, we discuss our classification principle: what kind of systems we are looking at, and what we mean by two systems are in the same class. In \refsec{sec-classification}, the classification principle is implemented mathematically, and some examples are discussed in \refsec{sec-example}. Finally, we investigate the instability in \refsec{sec-instability}.

\section{Principle of classification}\label{sec-setting}
A classification problem, formally speaking, is to classify elements of a set according to some equivalence relations. In many problems of condensed matter physics, the set is usually taken to be the set of Hamiltonians $H$ with an ``energy gap,'' while $H_1$ and $H_2$ are equivalent if and only if they can be continuously connected while keeping the gap open.

For Hermitian systems, there is no subtlety regarding the meaning of the gap, since all eigenvalues of a Hermitian Hamiltonian are real and the meaning of a gap on the real line is clear. For non-Hermitian systems (interacting or not), however, the eigenvalues can be complex. Therefore, the meaning of a ``gap'' needs to be further clarified \cite{Fu,symtop_2018,10.1088/2515-7639/ab4092}. 

Consider a non-interacting non-Hermitian system with translational invariance. Standard second quantization and band theory give rise to momentum-dependent one-body Hamiltonians $H(k)$. In this article, we will call $\{H(k)\}~(k\in\text{BZ, the Brillouin zone})$ a band structure, which contains information of both their spectrum $E_i(k)$ and associated eigenstates $\ket{\psi_i(k)}$. 

One has at least the following different notions of the gap:
\begin{itemize}
    \item \textbf{Line gap} \cite{symtop_2018,2019arXiv190207217B}. There exists a (maybe curved) line $l$ in the complex energy plane which separates the plane into two disconnected pieces. We require $E_i(k)\notin l$ for all $i$ and $k$, and both connected components have some spectral points in them.
    \item \textbf{Point gap} \cite{Gong_2018,symtop_2018,PhysRevB.99.125103,PhysRevB.99.235112}. $E_i(k)\neq 0$ for all $i$ and $k$. Here, 0 is a reference point which can be altered to any $E_0$.
    \item \textbf{Separable band} \cite{Fu,PhysRevB.98.155430_2018}. A specific band $E_i(k)$ is called separable if $E_j(k)\neq E_i(k)$ for all $j\neq i$ and $k$.
    \item \textbf{Isolated band} \cite{Fu}. A specific band $E_i(k)$ is called isolated if $E_j(k')\neq E_i(k)$ for all $j\neq i$ and $k,k'$.
\end{itemize}
Note that these notions are not mutually exclusive. For example, an isolated band is always separable; systems with isolated bands always have line gaps and hence always have point gaps. Also note that, the first two notions are applicable to general non-Hermitian systems, while the last two notions are specific to translational-invariant non-interacting cases by definition.

In our article, we will consider the classification of separable band structures, since other cases were solved \cite{Gong_2018,symtop_2018,PhysRevB.99.125103,PhysRevLett.123.066405_2019,PhysRevB.99.235112} by mapping back to the Hermitian case. However, there is one more problem with the definition of separability that needs to be discussed: the above-mentioned $E_i(k)$ may not be a well-defined function of $k$. 

For example, consider a one-dimensional systems with two bands, satisfying $E_1(k)\neq E_2(k)$ for all $k$. It is possible that $E_1(2\pi)=E_2(0)$: if one follows the spectrum when $k$ goes around the  Brillouin zone (a circle in this case) starting from $E_1(0)$, one may go to $E_2(0)$ instead of going back to $E_1(0)$; see \reffig{pic-z2braiding}. In this case, the notation ``$E_i(k)$'' (and therefore its separability) for a specific $i$ may not be well defined. Instead, it is better to define separability in a global manner: for any $k$, $E_i(k)~(i=1,\cdots,n)$ are all different. This definition of separability automatically rules out exceptional points, i.e., $H(k)$ is not diagonalizable under similarity transformations, since it requires (algebraically) degenerated spectra.
\begin{figure}
    \centering
    \includegraphics[width=0.9\columnwidth]{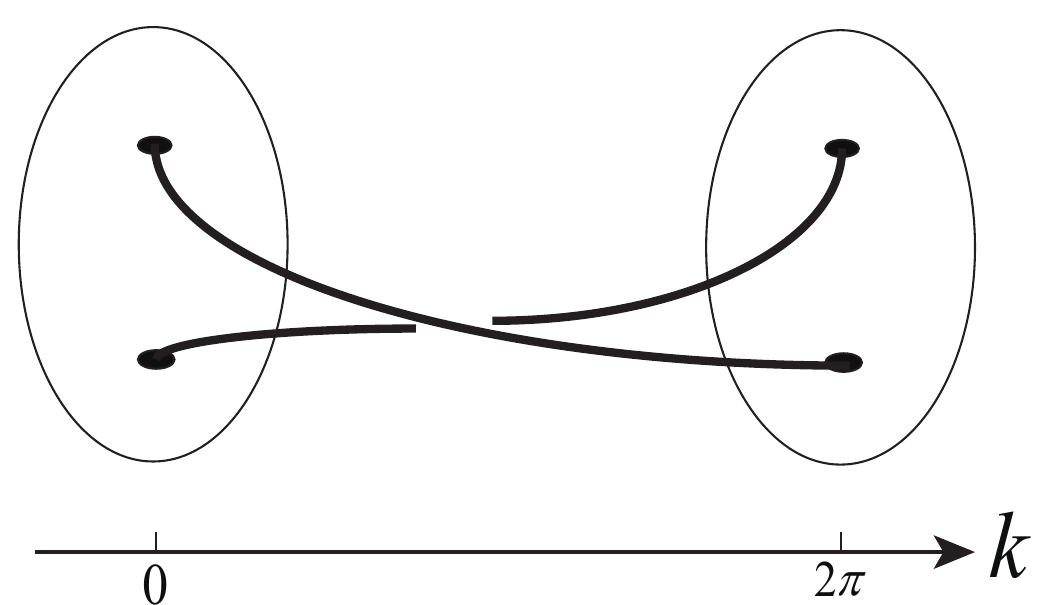}
    \caption{$\Z_2$ braiding of energy levels. In this figure, the disk is the complex energy plane, with two spectra points in it; $k$ is the Bloch momentum, $k=0$ and $k=2\pi$ should be identified. }\label{pic-z2braiding}
\end{figure}

To summarize, we will consider the following problem: classify the band structure $\{H(k)\}$ where spectra of $H(k)$ are non degenerated and $\{H_0(k)\}$ and $\{H_1(k)\}$ are equivalent if and only if they can be continuously connected  by $\{H_t(k)\}$ for $t\in[0,1]$ and the spectra of $H_t(k)$ for any $t$ and $k$ are always non-degenerated.

\section{classification}\label{sec-classification}

Let us consider the general problem of classifying band structures with $n$ bands on an $m$-dimensional lattice. Denote 
\begin{equation}
    X_n=\text{the space of~} H(k).
\end{equation}
Namely, it is the space of $n\times n$ matrices with non-degenerated spectrum. Here the Brillouin zone will be the $m$-dimensional torus $T^m$. Mathematically speaking, we want to find the homotopy equivalent classes of non-based maps from the Brillouin zone $T^m$ to $X_n$, denoted by $[T^m,X_n]$. 

It will be important to distinguish $T^m$ and $S^m$, since they will give different answers. It is also important to distinguish based maps and non-based maps: the former require a chosen point in $T^m$ to be mapped to a chosen point in $X_n$ while the latter have no such requirement\footnote{On the other hand, for continuous systems the appropriate choice the based map from $S^m$ to $X_n$, since the Brillouin zone here is $S^n$ with the requirement that infinite momentum maps to some fixed point \cite{KitaevK}.}.

To calculate the classification, we are going to use some standard methods in algebraic topology. For an introduction, see Ref.~\onlinecite{Hatcher}.

\subsection{The space $X_n$ and its homotopy groups}
An element $H$ of $X_n$ is an $n\times n$ matrix with non-degenerated spectrum, which can be represented as $(\lambda_1,\cdots,\lambda_n,\alpha_1,\cdots,\alpha_n)$. Here, $(\lambda_1,\cdots,\lambda_n)$ are ordered eigenvalues satisfying $\lambda_i\neq\lambda_j$, i.e., 
\begin{equation}
(\lambda_1,\cdots,\lambda_n)\in\Conf_n(\mathbb{C}),
\end{equation}
where $\Conf_n(\mathbb{C})$ is the configuration space of \textbf{ordered} $n$-tuples in $\mathbb{C}$. $(\alpha_1,\cdots,\alpha_n)$ are corresponding eigenvectors (up to complex scalar multiplications), which are linearly independent. Denote the space of linearly independent ordered $n$-vectors (up to scalar) in $\mathbb{C}^n$ as $F_n$. We have
\begin{equation}\label{eq-Fn}
F_n    \cong GL(n)/\mathbb{C}^{*n},
\end{equation}
since $GL(n)$ acts transitively on $F_n$ and the stabilizer group is $\mathbb{C}^{*n}$, where $\mathbb{C}^{*}=\mathbb{C}-\{0\}$, the group of nonzero complex numbers. Another way to understand this equation is to consider the columns of a $GL(n)$ matrix, which are ordered  $n$-vectors in $\mathbb{C}^n$, while ``up to scalar'' is taken care of by $n$ independent scalar multiplications $\mathbb{C}^{*n}$. The space $F_n$ is actually homotopic to the full flag manifold (the space of subspace sequences) of $\mathbb{C}^n$.

This representation has some redundancies: one can permute $(\lambda_i,\alpha_i)$ and get the same matrix $H$. Therefore,
\begin{equation}\label{eq-Xn}
X_n\cong (\Conf_n\times F_n)/S_n,
\end{equation}
where $S_n$ is the permutation group acting on $\Conf_n\times F_n$ as simultaneous permutations of $(\lambda_i,\alpha_i)$.

Consider $\pi_1(\Conf_n)$, whose elements are (equivalence classes of) paths in $\Conf_n$, which correspond to some pure braidings of $n$ mutually different points in $\mathbb{C}$. Here, ``pure'' means each point goes back to itself after the braiding. This is true since we are considering ordered $n$-tuples. Therefore,
\begin{equation}
\pi_1(\Conf_n)=\PB_n,
\end{equation}
where $\PB_n$ is the pure braid group (no permutation) of $n$ points \cite{intro_braid,Kauffman}. 

It turns out \cite{Fadell_Neuwirth_1962} that $\Conf_n=K(PB_n,1)$, the classifying space of the group $PB_n$. Therefore,
\begin{equation}\label{eq-EMspace}
\pi_m(\Conf_n)=0, ~~ m\geq 2.
\end{equation}

The homotopy groups $\pi_m(F_n)$ can be obtained by the long exact sequence of homotopy groups \cite{Hatcher}, based on the fibration \refeq{eq-Fn}. For $m=1$, we have:
\begin{equation}
\pi_1(\mathbb{C}^{*n})\to\pi_1(GL(n))\to\pi_1(F_n)\to \pi_0(\mathbb{C}^{*n})=0.
\end{equation}
Here, $\pi_1(\mathbb{C}^{*n})=\Z^n$, $\pi_1(GL(n))=\Z$, which is essentially the determinant. The map $\pi_1(\mathbb{C}^{*n})\to\pi_1(GL(n))$ is exactly summing over $n$ components in $\Z^n$, which is surjective. Therefore $\pi_1(F_n)=0$. For $m=2$, we have:
\begin{equation}\label{eq-seqm2}
0=\pi_2(GL(n))\to\pi_2(F_n)\xrightarrow{\partial}\pi_1(\mathbb{C}^{*n})\to\pi_1(GL(n)).
\end{equation}
Therefore, $\pi_2(F_n)$, as the kernel, is represented by $n$ integers with summation equal to 0:
\begin{equation}\label{eq-kernelrep}
    \{(t_1,\cdots,t_n)\in\Z^n|\sum t_i=0\},
\end{equation}
which is isomorphic to $\mathbb{Z}^{n-1}$. This representation with $n$ integers will be useful later. For $m\geq3$, we have
\begin{equation}
0=\pi_m(\mathbb{C}^{*n})\to\pi_m(GL(n))\to\pi_m(F_n)\to \pi_{m-1}(\mathbb{C}^{*n})=0,
\end{equation}
therefore $\pi_m(F_n)=\pi_m(GL(n))=\pi_m(U(n))$.

To summarize, the result is as follows:
\begin{equation}\label{eq-piF}
\pi_m(F_n)=\begin{cases}
0,& m=1,\\
\mathbb{Z}^{n-1},& m=2,\\
\pi_m(U(n)),& m\geq 3.\\
\end{cases}
\end{equation}

Now consider the space $X_n$. According to \refeq{eq-Xn} and the fact that $S_n$ is discrete, higher homotopy groups $\pi_{m\geq 2}(X_n)$ are the same as those of $\Conf_n\times F_n$, therefore the same as \refeq{eq-piF}, due to \refeq{eq-EMspace}. 

For the fundamental group $\pi_1$, one can take advantage of the fact that $\pi_1(F_n)=0$ and show that:
\begin{equation}
\pi_1(X_n)=\pi_1(\Conf_n/S_n)=B_n.
\end{equation}
Here $S_n$ acts on $\Conf_n$ by permutations, giving the configuration space $\Conf_n/S_n$ of \textbf{non ordered} $n$-tuples in $\mathbb{C}$, whose fundamental group is $B_n$, the braid group including ``non pure'' braidings.

This is because a loop in $X_n$ corresponds to a path $p(t)=(p_1(t),p_2(t))$ in $\Conf_n\times F_n$ such that $p_1(1)=gp_1(0)$ and $p_2(1)=gp_2(0)$ for the same $g\in S_n$. Note that $g$ is uniquely determined by $p_1(1)$ [the initial point $(p_1(0),p_2(0))$ is a fixed lifting] and $F_n$ is simply connected, the path one-to-one (homotopically) corresponds to a path in $\Conf_n$ with  $p_1(1)=gp_1(0)$ and therefore a loop in $\Conf_n/S_n$. A more algebraic proof is to note that the actions of $S_n$ on $\Conf_n\times F_n$ and $\Conf_n$ are consistent, which gives the pullback
\begin{equation}
\begin{tikzcd}
\Conf_n\times F_n \arrow[r] \arrow[d] & \Conf_n \arrow[d] \\
X_n \arrow[r]& \Conf_n/S_n
\end{tikzcd},
\end{equation}
and then apply the homotopy exact sequence for this pullback square.

The appearance of the braid group $B_n$ is easy to understand. Consider a one-dimensional band structure and follow the evolution of spectrum $\{E_i\}$ along the Brillouin zone circle. Similar to $n=2$ case in \refsec{sec-setting} as shown in \reffig{pic-z2braiding}, in general points in $\{E_i\}$ will braid with each other during this evolution and may become other points after one cycle. The evolution of $n$ disjoint points is topologically classified by the braid group $B_n$.

\subsection{The set $[T^m,X_n]$} \label{sec-topology}

The equivalent class $[T^m,X_n]$ is related but may not be equal to the homotopy group $\pi_m$, which is, by definition, $\braket{S^m,X_n}$. Here, $\braket{-,-}$ is used for based maps, while $[-,-]$ is used for non-based maps. In general, $[T,X]$ is just a set with no extra structures, even if $T$ is a sphere, in which case $\braket{T,X}$ is exactly a homotopy group. The relation between $[T,X]$ and $\braket{T,X}$ for general spaces $T$ and $X$ is as follows \cite{Hatcher}: There is a right action of $\pi_1(X)$ on $\braket{T,X}$, and $[T,X]\cong\braket{T,X}/\pi_1(X)$, the orbit set of the action.

We will first calculate $\braket{T^m,X_n}$ and then use the above connection to obtain $[T^m,X_n]$.

In the case $m=1$, $\pi_1(X_n)$ acts on $\braket{T^1,X_n}=\pi_1(X_n)$ by conjugate:
\begin{equation}
[f][\gamma]=[\gamma^{-1}\circ f\circ\gamma],
\end{equation}
therefore $[T^1,X_n]$ is the set of conjugacy classes of group $B_n$. Determining the conjugacy classes of braid group $B_n$ is a difficult problem\footnote{For a review of the conjugacy problem in braid groups, see Ref.~\onlinecite{braidbook}. It relies on the ``Garside structure''  \cite{Garside}.} except for $n\leq 2$. Geometrically they one-to-one correspond to equivalence classes of closed braids in the solid torus, which in turn can be regarded as (special) links ---collections of knots ---in the solid torus\footnote{Note that here the equivalence of links and knots is defined as isotopy inside the solid torus, instead of isotopy inside $\mathbb{R}^3$ (which is the usual meaning of equivalent links). For example, in Sec.~\ref{sec-example1D}, even if the eigenvalues have a nonzero ``spectral vorticity" as in Fig.~\ref{pic-z2braiding}, the associated knot is trivial in $\mathbb{R}^3$.}.

In the case $m=2$, the set $\braket{T^2,X}$ is given by \cite{mathstack} (see also Appendix \ref{app-torus})
\begin{equation}\label{eq-torus}
\bigcup_{\substack{a,b\in \pi_1(X)\\ab=ba}}\pi_2(X)/\langle t-t^a,t-t^b\mid t\in\pi_2(X)\rangle,
\end{equation}
where $t^a$ is the result of $a\in\pi_1(X)$ acting on $t\in\pi_2(X)$. Note that this is a noncanonical identification.  In our problem, the result is:
\begin{equation}\label{eq-2dbased}
\bigcup_{\substack{a,b\in B_n\\ab=ba}}\mathbb{Z}^{n-1}/\braket{t-t^a,t-t^b}\defeq \bigcup_{\substack{a,b\in B_n\\ab=ba}} Q(n,a,b).
\end{equation}
In other words, the classification of based maps is decomposed into several sectors, denoted by a pair of commuting braidings\footnote{For a review of the centralizer problem in the braid group, see Ref.~\onlinecite{braid_review}. The result heavily depends on the geometry of braiding, namely, the Nielsen-Thurston classification \cite{Nielsen-Thurston}.} $a,b\in B_n$; classification within each sector $(a,b)$ is given by the quotient  $Q(n,a,b)$, a finite-generated Abelian group, by identifying $t$ with $t^a$ and $t^b$. 

Physically, the braidings $a,b$ are given by following two nontrivial circles $l_a,l_b$ in the Brillouin zone $T^2$. Since $l_al_bl_{a}^{-1}l_{b}^{-1}$ is the boundary of the 2-cell of $T^2$, the corresponding braiding $aba^{-1}b^{-1}$ must be trivial, hence $ab=ba$. Fixing $a,b$, the map on the 2-cell is determined by $\pi_2(X_n)=\pi_2(F_n)=\mathbb{Z}^{n-1}$, which are essentially $(n-1)$ Chern numbers, up to some ambiguities taken care of by the quotient.

The action $t\mapsto t^a$ here is determined as follows. Recall from \refeq{eq-kernelrep} that $\pi_2(X_n)=\Z^{n-1}$ can be represented by $\{(t_1,\cdots,t_n)\in\Z^n|\sum t_i=0\}$. $a\in B_n$ induced a permutation $\tilde{a}\in S_n$ by forgetting the braiding. Then $t^a$ is represented by a permutation of $(t_1,\cdots,t_n)$:
\begin{equation}\label{eq-permutation}
    (t_1,\cdots,t_n)\mapsto (t_{a(1)},\cdots,t_{a(n)}).
\end{equation}
The proof of this statement is a bit technical. However, since it is the root of most novel classifications in this article, we give a detailed proof in Appendix \ref{app-permutation}.

Now consider the action of $\pi_1(X_n)=B_n$ on $\braket{T^2,X_n}$. Pick $c\in B_n$; then $c$ act on $(a,b)$ by conjugate:
\begin{equation}\label{eq-pairconju}
(a,b)\to (c^{-1}ac,c^{-1}bc).
\end{equation}
The action of $c$ on $\bar{t}\in Q(n,a,b)$ is induced by the action of $\pi_1(X_n)$ on $\pi_2(X_n)$: under $c$, $t$ goes to $t^c$, $t^a$ goes to $t^{ac}$, therefore $t-t^a$ goes to $t^c-(t^c)^{c^{-1}ac}$, therefore the action of $c$ on $\bar{t}\in Q(n,a,b)$ is well-defined as $\bar{t}^c=\overline{t^c}\in Q(n,c^{-1}ac,c^{-1}bc)$. Note that $Q(n,a,b)\cong Q(n,c^{-1}ac,c^{-1}bc)$, due to fact that \refeq{eq-2dbased} and \refeq{eq-permutation} only care about the permutation structure of $a,b$, which is invariant under conjugation. We finally get:
\begin{equation}\label{eq-result2d}
[T^2,X_n]=\bigcup_{\substack{\overline{(a,b)}\\ab=ba}}\overline{Q}(n,a,b),
\end{equation}
where $\overline{(a,b)}$ means a conjugacy class of commuting pairs under \refeq{eq-pairconju}, and $\overline{Q}(n,a,b)=Q(n,a,b)/\pi_1^s(X_n)$ is the orbit set (not quotient group) of $Q(n,a,b)$ under the stabilizer subgroup $\pi_1^s(X_n)$ that keeps $(a,b)$ invariant.

The reason for the appearance of this $\pi_1^s(X_n)$ action can be traced back to the difference between $[T,X]$ and $\braket{T,X}$. Physically, there is no natural way to label the bands (even if no braiding happens, namely, $a=b=\text{id}$). In the Hermitian case, bands are naturally ordered according to their energy, which is not the case for complex energy levels.  Therefore there are some redundancies corresponding to change the label of bands (see \refsec{sec-2DChern} for an example). Also note that, while $Q(n,a,b)$ is a finite generated Abelian group, $\overline{Q}(n,a,b)$ is just a set.   

\section{examples}\label{sec-example}
\subsection{Non-Hermitian bands in one dimension}\label{sec-example1D}
In the case of $m=1$, we know from \refsec{sec-topology} that band structures are classified by the conjugacy classes of group $B_n$. 

Determining the conjugacy classes of braid group $B_n$ is only easy when the number of bands $n=2$, where the braid group $B_2$ is just $\Z$: $a\in\Z$ is the number of elementary braidings (half of a $2\pi$ rotation), with $a$ even implying a pure braid and $a$ odd implying a permutation. In this case, each conjugacy class only contains one element, since $\Z$ is Abelian. Therefore, the classification is given by:
\begin{equation}
[T^1,X_2]=\Z.
\end{equation}
The same classification was found in Ref.~\onlinecite{Fu}; see Eq.(8) therein. Note that authors there use $\frac{1}{2}\Z$ instead of $\Z$: their spectral ``vorticity'' is exactly half of the above invariant.

\subsection{Two-band Chern ``insulators''}\label{sec-2DChern}
Consider the case with $m=n=2$, namely, band structures with two bands in two-dimensional (2D) space. This corresponds to Chern insulators in the Hermitian case. However, it may not be a true insulator in the non-Hermitian case if there is no line gap (to place the chemical potential).

Let us calculate $Q(2,a,b)$ and $\overline{Q}(2,a,b)$, where $a,b\in B_2=\Z$. There are four cases, depending on the even/odd values of $a$ and $b$.
\begin{itemize}
    \item $a,b$ even. Then $t^a=t^b=t$, therefore $Q(2,a,b)=\Z$. The action of $c\in\pi_1(X_n)$ on $Q(2,a,b)$ might be nontrivial: it acts as taking opposite if $c$ is odd (see below), therefore $\overline{Q}(2,a,b)=\mathbb{N}$, the set of nonnegative integers.
    \item $a$ even, $b$ odd. Then $t^a=t$ while $t^b=-t$ in the sense that $(s,-s)^b=(-s,s)$. Therefore $Q(2,a,b)=\braket{(s,-s)}/\braket{(2s,-2s)}\cong\Z_2$. The action of $\pi_1(X_n)$ at most takes $(s,-s)$ to $(-s,s)$, which has no effects on $\Z_2$, therefore $\overline{Q}(2,a,b)=\Z_2$.
    \item $a$ odd, $b$ even. Same as above.
    \item $a,b$ odd. Then $t^a=t^b=-t$, $Q(2,a,b)=\overline{Q}(2,a,b)=\Z_2$.
\end{itemize}
Therefore, band structures are classified by:
\begin{equation}
\bigcup_{a,b\in \Z} \N \text{~or~} \Z_2.
\end{equation}

\subsubsection{Understanding the $\N$ invariant}
The $\N$ classification (instead of $\Z$) comes from the fact that we have no natural way to identify ``upper band'' and ``lower band'' as in the Hermitian case, since there $\C$ is not naturally ordered as $\R$. This new feature of non-Hermitian classification will disappear if, for example, we have a \textbf{fixed} line gap, where the classification will go back to $\Z$.

\subsubsection{Understanding the $\Z_2$ invariant}
The $\Z_2$ classification in some sectors is a more interesting phenomenon. It comes from the interplay between spectrum braiding and eigenvector topology (Chern band). Here, we provide a formula as well as heuristic arguments for this invariant.

We will concentrate on the case where $(a,b)$ is $(\text{even},\text{odd})$. The $(\text{odd},\text{even})$ case is similar; the $(\text{odd},\text{odd})$ case can be handled by a Dehn twist\footnote{Cut the torus along a longitude, resulting in a cylinder; then gradually twist the cylinder so that one boundary is fixed and the other boundary rotates $2\pi$; then glue it back. Using Dehn twist, one can reduce $(\text{odd},\text{odd})$ to $(\text{even},\text{odd})$.}. Also note that the $\Z_2$ invariant essentially comes from the $(\text{odd},\text{even})$ sector of $[T^2,F_2/S_2]$ ($F_2/S_2$ is the space of distinct pairs of states), as one can see by following the same calculations as above.

We claim the following formula for this invariant:
\begin{equation}\label{eq-Chern}
    C=\frac{1}{2\pi}\int_{\text{BZ}}\epsilon_{ij}B_{ij}(k)d^2k+S_{WZW}(a,a').
\end{equation}
In the first term, the Berry curvature $B_{ij}(k)$ \cite{Fu} is defined by following one band, therefore it has a discontinuity at the boundary; the integral is over the conventional $2\pi\times 2\pi$ Brillouin zone. The second term $S_{WZW}(a,a')$ is a boundary Wess-Zumino-Witten (WZW) term\footnote{The boundary term is also known as a Berry phase term \cite{Fradkin}.}, defined as follows. By following one band, one gets a map from a cylinder to the Bloch sphere $S^2$, such that for any point on the left boundary $a$, the corresponding point on the right boundary $a'$ maps to a different point on $S^2$ (since they correspond to linear independent vectors); see Fig.~\ref{pic-Chern}. We then close two boundaries in a \textbf{consistent} way \cite{MooreBalents} such that the above condition is still satisfied on two ``caps." This is always possible since $[\tilde{a}]=0$ in $\pi_1(F_2/S_2)$ due to the assumption that $a$ is even, where $\tilde{a}$ denotes the map from the nontrivial loop (boundary) to the space of pairs $F_2/S_2$. Then $S_{WZW}(a,a')$ is defined as
\begin{equation}
    \frac{1}{4\pi}\times \text{oriented area of caps on~} S^2.
\end{equation}
\begin{figure}
    \centering
    \includegraphics[width=1\columnwidth]{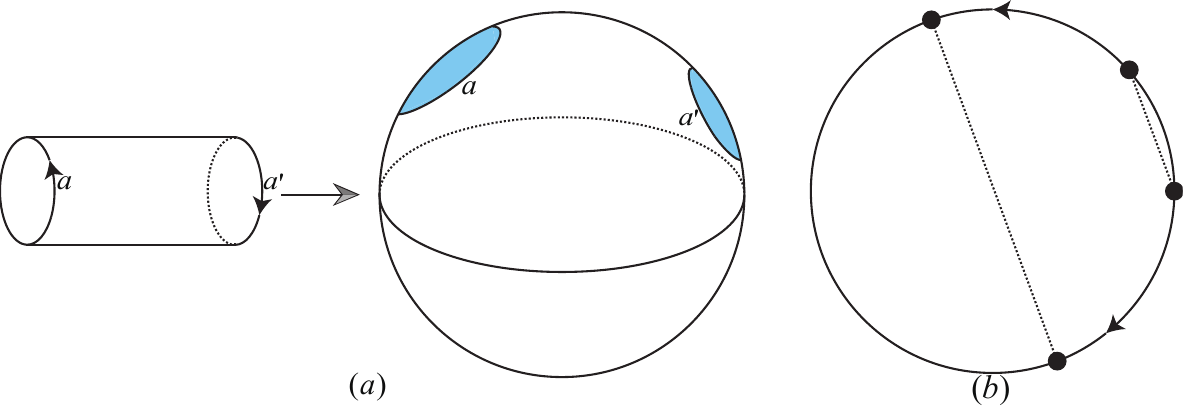}
    \caption{(a) Following one band, we get a map from cylinder to $S^2$. To define the WZW term, one needs to close the cylinder with two caps in a consistent way (not necessarily antipodal in the non-Hermitian case). Here the arrow represents orientation, not to be confused with the way one identifies two boundaries. (b) A deformation retraction from  $F_2/S_2$ to $RP^2$. Each point in $F_2/S_2$ corresponds to a pair of different points in $S^2$. We draw the great circle corresponding to that pair, and then gradually push the pair to an antipodal pair (corresponding to a point in $RP^2$). If a pair is already antipodal, nothing needs to be done. In this way, we have defined a deformation retraction. }\label{pic-Chern}
\end{figure}

By adding the caps, we obtain a closed manifold, therefore \refeq{eq-Chern} is an integer. As always, there are some ambiguities in the definition of $S_{WZW}$, corresponding to the ambiguities in adding the caps. Importantly, the consistency  for the caps requires that \refeq{eq-Chern} can only be shifted by 2 (instead of 1) by the ambiguities\footnote{Heuristically, one cannot just flip one cap while leaving the other unflipped in Fig. \ref{pic-Chern}(a), otherwise the Brouwer's fixed point theorem guarantees an inconsistent point.}. To see this, note that we have a deformation retraction from $F_2/S_2$ to $RP^2$, as defined in Fig. \ref{pic-Chern}(b). After this deformation retraction, the consistency condition simply requires that corresponding points in two boundaries map to antipodal points. Therefore, two caps should always be antipodal to each other, and 
\begin{equation}
    S_{WZW}(a,a')=2S_{WZW}(a).
\end{equation}
Therefore, $C$ is only defined mod 2 and we obtain a $\Z_2$ invariant.

Another way to understand the $\Z_2$ invariant is as follows. Using the above deformation retraction, we see that this $\Z_2$ can also be understood from $[T^2,RP_2]$. For the $(\text{odd},\text{even})$ sector, we have the diagram
\begin{equation}
\begin{tikzcd}
2T^2 \arrow[r] \arrow[d] & S^2 \arrow[d] \\
T^2 \arrow[r]& RP^2
\end{tikzcd},
\end{equation}
therefore the classification amounts to classifying covariant maps from $2T^2$ to $S^2$. Here, $2T^2$ is a double cover of Brillouin zone $T^2$, by gluing two cylinders along the $b$ direction; see \reffig{pic-z2invariant}; ``covariant" means corresponding points in the left and right cylinder should map to antipodal points. 

Now, we can add a ``bump'' of Berry curvature with positive 1 integral and a ``bump'' with negative 1 integral by deforming the eigenstates, both in the ``upper band.'' It is necessary to add opposite bumps due to the covariant constraint.  We can then move a pair of bumps along the $b$ direction for $2\pi$. After this procedure, we effectively add a $C=2$ bump to the  ``upper band'' and a $C=-2$ bump to the ``lower band.'' Therefore, $C$ is again only well defined mod 2.

\begin{figure}
    \centering
    \includegraphics[width=1\columnwidth]{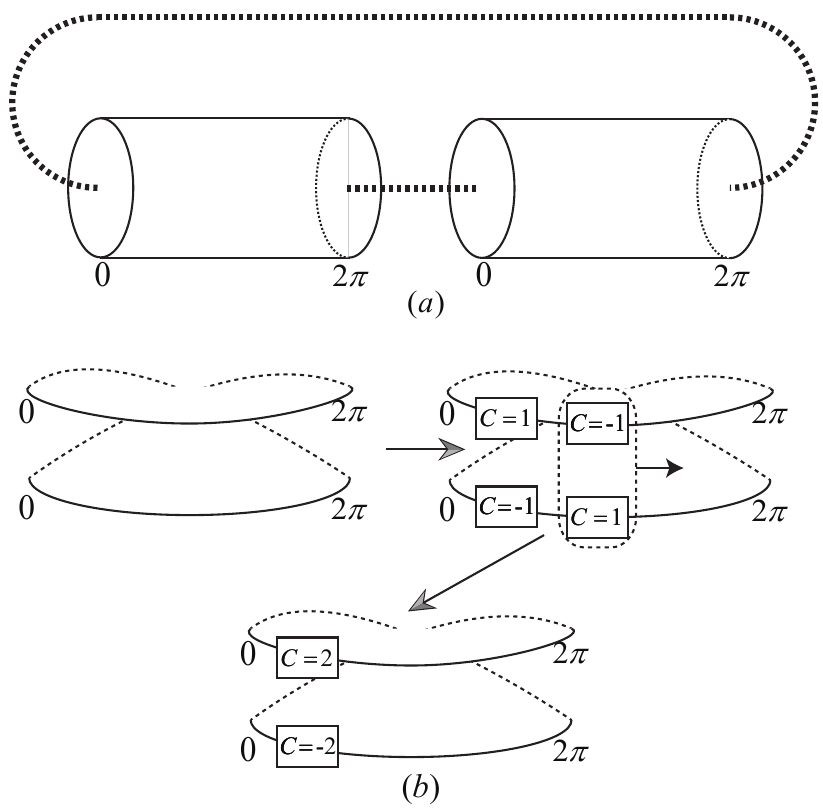}
    \caption{Physical origin of the $\Z_2$ invariant. (a) Due to energy level braiding, the Brillouin zone is better considered as a torus with size $4\pi\times 2\pi$ on which the energy $E(k)$ and wave function $\psi(k)$ is well defined: if one follows one band on the conventional Brillouin zone of size $2\pi\times 2\pi$, then after $2\pi$ one goes to the other band. The dashed lines indicate this ``gluing.'' (b) Each solid line represents a cylinder similar to those in (a), the dashed lines again indicate the ``gluing.'' Starting with trivial bands, one adds $\pm 1$ ``bumps'' to the upper band and $\mp 1$ ``bumps'' to the lower band, then moves the right pair (circled by the dashed line) to the right. After $2\pi$, they will switch, resulting in a $C=2$ ``bump'' in the upper band and $C=-2$ ``bump'' in the lower band. }\label{pic-z2invariant}
\end{figure}

\subsection{$n$-band Chern ``insulators''}\label{sec-nDChern}
The conjugacy classes and commuting pairs are hard to describe if $n>2$. However, the quotients $Q(n,a,b)$ for given braiding sector $(a,b)$ are not hard to calculate.

Recall from \refeq{eq-permutation} that $\pi_2(X_n)=\Z^{n-1}=\{(t_1,\cdots,t_n)|\sum t_i=0\}$ and $(t_1,\cdots,t_n)^a=(t_{a(1)},\cdots,t_{a(n)})$ where $a(i)$ is the image of $i$ under the permutation $a$. Therefore, the subgroup to be quotient out is (we only write down the $t^a$ part):
\begin{equation}\label{eq-subgroup}
\begin{aligned}
\braket{t-t^a}=\braket{(t_1-t_{a(1)},\cdots,t_n-t_{a(n)})|\sum t_i=0}.
\end{aligned}
\end{equation}

There are only $(n-1)$ independent $t_i$: we can use $t_n=-\sum_{i=1}^{n-1}t_i$ to get rid of the constraint. Then the subgroup \refeq{eq-subgroup} is generated by $2(n-1)$ (not necessarily independent) generators. For example, take $t_1=1, t_2=\cdots=t_{n-1}=0, t_n=-1$ and consider the action of $a$, we get a generator $e_1-e_{a^{-1}(1)}-e_n+e_{a^{-1}(n)}$, where $e_1=(1,0,\cdots,0,-1), e_2=(0,1,0,\cdots,0,-1),e_{n-1}=(0,\cdots,0,1,-1), e_n=(0,\cdots,0)$. Therefore, the $Q(n,a,b)$ is the quotient of $\braket{e_1,\cdots,e_{n-1}}$ with $2(n-1)$ relations $e_i-e_{a^{-1}(i)}-e_n+e_{a^{-1}(n)} (i=1,\dots,n-1)$. Its structure can be determined by standard procedure using the Smith normal form (normal form for integer matrix under elementary row/column operations).

As a simple example, consider the case where $a: 1\to 2\to 3\to 4\to 1$, $b: 1\to 4\to 3\to 2\to 1$. This is possible, say, by taking $a$ to be a braiding with such permutation structure, then taking $b=a^{-1}$. The auxiliary ``generators'' given by $t_i=1, t_1=\cdots=t_{i-1}=t_{i+1}=\cdots=t_n=0$ written in terms of $n$-tuples are the $i$th columns of the following matrix
\begin{equation}\label{eq-cycmatrix}
\begin{bmatrix}
1&-1&&\\
&1&-1&\\
&&1&-1\\
-1&&&1
\end{bmatrix}.
\end{equation}
Since the true generators are given by taking $t_i=1, t_n=-1, t_j=0 (j=1,\cdots,i-1,i+1,\cdots,n-1)$, we need to subtract the last column from all other columns and delete the last row. The matrix of true generators for $\braket{t-t^a}$ is:
\begin{equation}\label{eq-genmatrix}
\begin{bmatrix}
1&-1&\\
&1&-1\\
1&1&2\\
\end{bmatrix},
\end{equation}
and similarly for $b$:
\begin{equation}
\begin{bmatrix}
2&1&1\\
-1&1&\\
&-1&1\\
\end{bmatrix}.
\end{equation}
Juxtaposing those two matrices and calculating the Smith normal form of the result, we get
\begin{equation}
\begin{bmatrix}
1&0&0&0&0&0\\
0&1&0&0&0&0\\
0&0&4&0&0&0\\
\end{bmatrix},
\end{equation}
which means $Q(4,a,b)=\Z_4$.

Another example is when $b=\text{id}$, i.e., no permutation. We can decompose $a$ into cycles: $a=(...)(...)...(...)$. Denote the length of each cycle to be $l_1,\cdots,l_k$  $(\sum_{i=1}^k l_i=n)$ where $k$ is the number of cycles. In this case, we can follow the above procedure and get an explicit formula for $Q(n,a,b)$.

Denote an $l\times l$ matrix of form \refeq{eq-cycmatrix} to be $J_l$; then the counterpart of \refeq{eq-cycmatrix} (where columns are  auxiliary ``generators'') is:
\begin{equation}\label{eq-blockJ}
\begin{bmatrix}
J_{l_1} &&&\\
&J_{l_2}&&\\
&&\ddots&\\
&&&J_{l_k}
\end{bmatrix},
\end{equation}
and the counterpart of \refeq{eq-genmatrix} by subtraction and deleting is:
\newcommand{\blocked}[1]{\multicolumn{2}{c|}{\smash{\raisebox{.5\normalbaselineskip}{$\mathlarger{\mathlarger{#1}}$}}}}
\newcommand{\blockedd}[1]{\multicolumn{2}{c}{\smash{\raisebox{.5\normalbaselineskip}{$\mathlarger{\mathlarger{#1}}$}}}}
\begin{equation}
\left[\begin{array}{cc|cc|cc|cc}
&         &      &       &&&&\\
\blocked{J_{l_1}}&  &    &&&&   \\
\hline
 &   &       &  &&&&\\
 &        & \blocked{J_{l_2}}&&&&\\
 \hline
 &&&&&&&\\
  &&&&\blocked{\ddots}&&\\
 \hline
 &&&&&&\\
 1&1&1&1&1&1&\blockedd{\bar{J}_{l_k}}
\end{array}\right],
\end{equation}
where $\bar{J}_{l_k}$ is an $(l_k-1)\times (l_k-1)$ matrix of form \refeq{eq-genmatrix}. To clarify, the last row of the above big matrix is $(1,1,\dots,1,2$). It is easy to perform row transformation on the above matrix and get:
\begin{equation}
\left[\begin{array}{cc|cc|cc|cc}
&         &      &       &&&&\\
\blocked{K_{l_1}}&  &    &&&&   \\
\hline
&   &       &  &&&&\\
&        & \blocked{K_{l_2}}&&&&\\
\hline
&&&&&&&\\
&&&&\blocked{\ddots}&&\\
\hline
&&&&&&\\
0&l_1&0&l_2&\cdot\cdot&\cdot\cdot&\blockedd{\overline{K}_{l_k}}
\end{array}\right],
\end{equation}
where $K_l=\text{diag}\{1,\cdots,1,0\}$ (size $l$), $\overline{K}_l=\text{diag}\{1,\cdots,1,l\}$ (size $l-1$). To clarify, the last row of the above big matrix is $(0,\cdots,0,l_1,0,\cdots,0,l_2,\cdots,0,\cdots,0,l_k)$. Therefore, the structure of $Q(n,a,b)$ is:
\begin{equation}\label{eq-gcd}
Q(n,a,b)=\Z^{k-1}\oplus\Z_{\gcd{(l_1,\dots,l_k)}},
\end{equation}
where $\Z_{\gcd{(l_1,\dots,l_k)}}$ is the greatest common divisor and $\Z_1$ means trivial group $\{0\}$ if $\gcd=1$.

The $\Z^{k-1}$ comes from the fact that we have $k$ groups of bands (bands that transfer to each other under braidings are in the same group). Each band has an integer Chern number, with summation equal to 0. This is the same as the Hermitian case. However, there is an extra $\Z_{\text{gcd}}$. We also see that the extra torsion part is determined by \textbf{all band groups} as a whole, not from any specific band group. It shows some complicated interplay between energy braiding and eigenstates topology.

With other permutations $a,b (a,b\neq \text{id})$, it is possible to get more than one torsion. An example is $a: 1\leftrightarrow 2, 3\leftrightarrow 4$ with $b: 1\leftrightarrow 3, 2\leftrightarrow 4$. The algorithm will give us $\Z_2\oplus\Z_2$.

\section{instability}\label{sec-instability}
Examples in \refsec{sec-example} show that our homotopical approach reveals more topological invariants than the traditional $K$-theory approach. For example, a two-band Chern ``insulator'' in 2D may reveal some $Z_2$ classification due to the nontrivial topology of the spectrum.

Similar phenomena also happen in the Hermitian world. For example, in three dimensions (3D), insulators in class $A$ are always trivial according to the periodicity table. However, one can still have a $\Z$ classification if the number of bands is fixed to be 2, due to $\pi_3(\mathbb{C}P^1)=\Z$. This is called the Hopf insulator \cite{Hopfinsulator}, which is \textbf{unstable} against adding more bands. Indeed, as long as one adds one more band above and below the Fermi surface respectively, the classification will be trivial due to $\pi_3(Gr_{\mathbb{C}}(3,1))=0$ (and similarly for more bands), where $Gr_{\mathbb{C}}(3,1)$ is the complex Grassmannian.

A natural question arises: Are our new topological invariants stable against adding bands?

As an example, let us consider two-band systems in 2D as in \refsec{sec-2DChern}, and add one more band. Since the classification is decomposed into braiding sectors and each sector has its own classification set, it only makes sense to add a band with no permutation with previous bands (therefore it does not alter the braiding sectors). For each sector $(a,b)$, adding a band without permutation is to add a length-$1$ cycle after previous $a,b$, denoted by $a', b'$.
\begin{itemize}
    \item $a,b$ even. Then $a'$ and $b'$ are trivial permutations, therefore $Q(3,a,b)=\Z^2$, which are just two Chern numbers.
    \item $a$ even, $b$ odd. Then is $a'$ trivial while $b'$ decomposes as $(1 2)(3)$. \refeq{eq-gcd} shows that $Q(3,a',b')=\Z$.
    \item $a$ odd, $b$ even. Same as above.
    \item $a,b$ odd. Then both $a'$ and $b'$ are of the form $(1 2)(3)$. A Smith normal form calculation shows $Q(3,a',b')=\Z$.
\end{itemize}

In all cases, we see that the extra band contributes a Chern number $\Z$, as well as kills the old $\Z_2$ invariants if there are any, even if the $\Z_2$ comes from other bands that never intersect with the added band. This is possible since the $\Z_2$ comes not just from those two bands, but from all three bands as a whole, as noted at the end of \refsec{sec-nDChern}.

The instability of $\Z_2$ can be understood as follows. Assuming $a$ odd and $b$ even, consider the procedure shown in \reffig{pic-instability}: we start with three bands with Chern numbers $1,-1,0$, where the first two bands switch to each other after $2\pi$ as in \reffig{pic-z2invariant}(a). Add a negative bump and a positive bump in band 1, as well as a positive bump and a negative bump in band 3; then move the rightmost bump pair in band 1 and 3, so that the positive bump cancels the negative bump in band 2; the remaining bumps in band 1 and 3 can be easily canceled, leaving three trivial bands. During the procedure, the local neutral condition is always satisfied. Note that the third band is essential for this argument to work.

\begin{figure}
    \centering
    \includegraphics[width=\columnwidth]{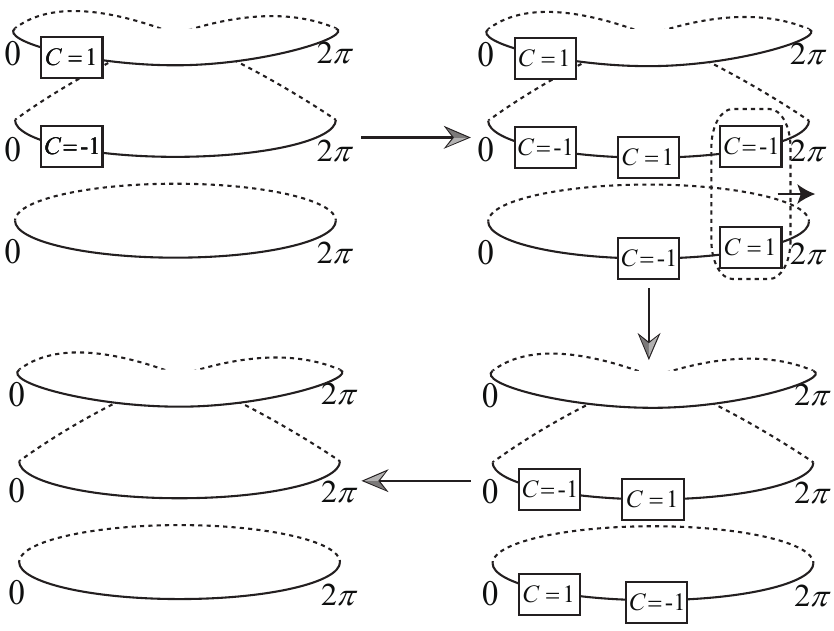}
    \caption{Physical origin of instability. Similar to \reffig{pic-z2invariant}, solid lines represent the Brillouin; dashed lines indicate the ``gluing.'' We start with three bands with Chern numbers $1,-1,0$ respectively, where the first two bands switch to each other after $2\pi$. If we forget band 3, it will be nontrivial, indicated by a $\Z_2$ invariant. However, adding band 3 and follow the procedure shown in the figure, we can make the band structure totally trivial.}\label{pic-instability}
\end{figure}

Similarly, as long as $b=\text{id}$, \refeq{eq-gcd} shows that $Q(n+1,a',b')$ has no torsion part.

We can prove a general result regarding the instability, even if $b\neq\text{id}$. For a system with $n$ bands, consider the braiding sector labeled by the commuting pair $(a,b)$. Let us add an extra no-permutation band; then the matrix of auxiliary ``generators'' is
\begin{equation}
\left[\begin{array}{cc|cc|c}
&         &      &       &0\\
\blocked{A}&  \blocked{B}    &0   \\
\hline
0&0   &0       &0  &0
\end{array}\right],
\end{equation}
where $A$ and $B$ are of form \refeq{eq-blockJ} up to some congruent transformation by permutation matrices. The matrix of generators [counterpart of \refeq{eq-genmatrix}] is therefore just
\begin{equation}
\mathlarger{\mathlarger{[A~|~B]}}.
\end{equation}
We claim that the invariant factors in its Smith normal form must be 1. Indeed, we claim a more general statement: 
\begin{claim*}
    Assume a matrix has the following property: there are either 2 or 0 nonzero elements in each column; in the former case, there is exactly one 1 and one $-1$. Then the invariant factors of this matrix must be 1 (if there are any).
\end{claim*}
\begin{proof}
    We prove by induction on the total number of nonzero elements $N$. From the assumption, $N$ must be even. If $N=0$, then the statement is trivially true.
    
    Now assume the state is true for $N$ and less, let us consider $N+2$. Denote the matrix to be $A$. Without lose of generality, assume $A_{1,1}=1$, $A_{2,1}=-1$, $A_{i,1}=0$ for $i\geq 3$. Add row 1 to row 2, denote the new matrix as $A'$, then $A'_{i,1}=0$ for $i\geq 2$. Moreover, from the assumption on matrix $A$, there are only seven possibilities happened to $\begin{bmatrix} A_{1,2} \\ A_{2,2} \end{bmatrix} $ (we only write four of them, the other three are obtained by adding negative signs):
    \begin{equation}
    \begin{bmatrix} 0 \\ 0 \end{bmatrix} \to \begin{bmatrix} 0 \\ 0 \end{bmatrix},
    \begin{bmatrix} 1 \\ -1 \end{bmatrix} \to \begin{bmatrix} 1 \\ 0 \end{bmatrix},
    \begin{bmatrix} 1 \\ 0 \end{bmatrix} \to \begin{bmatrix} 1 \\ 1 \end{bmatrix},
    \begin{bmatrix} 0 \\ 1 \end{bmatrix} \to \begin{bmatrix} 0 \\ 1 \end{bmatrix}.
    \end{equation}
    Therefore, the column $(A'_{2,2},\cdots,A'_{n,2})^T$ satisfies the same assumption as columns of $A$. The number of nonzero elements in this shorter column is less than or equal to that in   $(A_{1,2},\cdots,A_{n,2})^T$. Other columns are similar.
    
    Then we use column transformations to make $A'_{1,i}=0$ for $i\geq 2$, while keeping other elements. $A'$ is of the form:
    \begin{equation}
    A'=\left[\begin{array}{c|cc}
    1& 0 &0   \\
    \hline
    0&&\\
    0&  \blockedd{A''}
    \end{array}\right].
    \end{equation}
    We can then apply the induction assumption on $A''$ and finish the proof.
\end{proof}

Therefore, $Q(n+1,a',b')$ is always of the form $\Z^k$. This means all torsion invariants $\Z_i (i\geq 2)$ are unstable against adding a no-permutation band.

\section{conclusion and outlook}
In this article, we considered the homotopical classification of non-Hermitian band structures from first principles. We found that the whole classification set is decomposed into several sectors, based on the braiding of energy levels. Fix a braiding pattern, we consider the classification coming from nontrivial eigenstates topology. Since different bands will transfer to each other under braidings, the classification of band topology is not just a direct summation of Chern numbers. Instead, the interplay between energy level braiding and eigenstates topology gives some new torsion invariants.

The torsion invariants come from all bands as a whole, instead of some specific band group. Namely, even if we add a band with no crossing with previous bands, the torsion invariants can in principle be changed. We found that the torsion invariants are unstable, in the sense that just adding a trivial band can trivialize them. This statement is proved based on an interesting combinatorial argument.

Many future works can be done in this framework. First of all, due to the complexity of the braid group, it is complicated to describe its conjugacy classes and commuting pairs, let alone the conjugacy classes of commuting pairs. It will be useful to develop more explicit descriptions of the braiding sectors. On the more physical side, it is very interesting to consider the physical consequence of these novel invariants, in terms of physical observables. On the other hand, in this article we only consider the case with no symmetry as an first step. It is necessary to consider other symmetry classes using our framework in order to find a complete classification.

\textit{Note added}: after this work was presented on arXiv, Ref.~\onlinecite{lol} appeared and has partial overlap with our work. We thank the authors for bringing it to our attention and some helpful discussions thereafter.

\begin{acknowledgements}
This work is supported by NSF Grant No.~DMR-1848336. ZL is supported by the Mellon fellowship and PQI fellowship at University of Pittsburgh. ZL is grateful for the hospitality of the Visiting Graduate Fellowship program at Perimeter Institute where part of this work was carried out. Research at Perimeter Institute is supported in part by the Government of Canada through the Department of Innovation, Science and Economic Development Canada and by the Province of Ontario through the Ministry of Economic Development, Job Creation and Trade.

\end{acknowledgements}

\appendix
\section{Some algebraic topology details}

\subsection{Homotopy class $\braket{T^2,X}$: proof of \refeq{eq-torus} }\label{app-torus}

In this section, we prove \refeq{eq-torus} in detail:
\begin{equation}\label{eq-torusapp}
\begin{aligned}
&\braket{T^2,X}\\
=&\bigcup_{\substack{a,b\in \pi_1(X)\\ab=ba}}\pi_2(X)/\langle t-t^a,t-t^b\mid t\in\pi_2(X)\rangle,
\end{aligned}
\end{equation}
Namely, a homotopy class $[f]\in \braket{T^2,X}$ one-to-one corresponds to an element in the set described by the right-hand side of \refeq{eq-torusapp}.

For each pair $(a,b)\in\pi_1(X)^2$ such that $ab=ba$, we choose and fix two loops\footnote{Note the notations here: $a$ and $b$ are homotopy classes of loops, $a_0$ and $b_0$ are loops.} $a_0, b_0$, and also choose and fix a homotopy from $a_0b_0a_0^{-1}b_0^{-1}$ to 0, denoted by $F(a,b)$. Note that $a_0,b_0,F(a,b)$ are arbitrarily chosen. But once they are chosen, they are fixed for all.

Given $[f]\in \braket{T^2,X}$, we choose a map $f:T^2\to X$ in this class. There are two nontrivial loops (fixed) in $T^2$, denoted by $l_1,l_2$, with the same base point. The restriction of $f$ on $l_1$ defines a map (loop) $f_1:S^1\to X$ and therefore an element $a\in \pi_1(X)$. Similarly we have $b\in\pi_1(X)$. Since the loop $l_1l_2l_1^{-1}l_2^{-1}$ is homotopic to 0 in $T^2$, we know $ab=ba$ in $\pi_1(X)$. Obviously $a$ and $b$ are well-defined functions of $[f]$. 

 Since loop $f_1$ is homotopic to $a_0$, there exists (not unique) a homotopy $F(a)$ from $f_1$ to $a_0$; the same for $b$ and we have a $F(b)$. Now define an element $h$ in $\pi_2(X)$ as in \reffig{pic-torus}(a). In this way we get an element $\bar{h}\in \pi_2(X)/\langle t-t^a,t-t^b\mid t\in\pi_2\rangle$.
 \begin{figure}
    \centering
    \includegraphics[width=0.9\columnwidth]{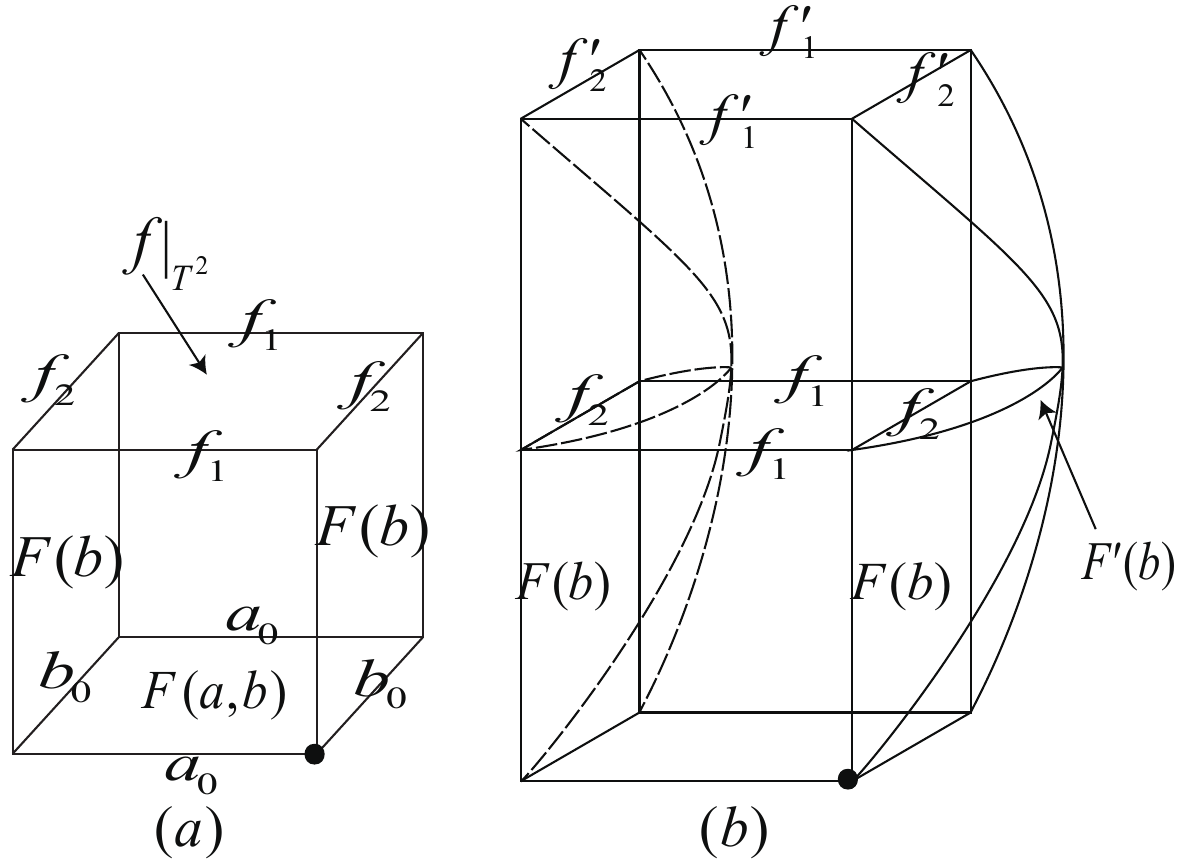}
    \caption{(a) The definition of $h\in\pi_2(X)$. It is defined by the surface of the cube, made of $f|_{T^2},F(a),F(b)$ and the pre-chosen $F(a,b)$. (b) $\bar{h}$ is well-defined. The ``cap" attached to the right surface (the other one in the left denoted with dashed line) is the other homotopy $F'(b)$. Here, $F(b), F'(b)$ and the homotopy between $f_2$ and $f'_2$ defined an element $s\in\pi_2(X)$. Pictorially it is the right ``cap"+right surface. The one on the left corresponds to $s^a$, since we have to fixed a base point when defining $\pi_2(X)$, say the one denoted by $\bullet$. }\label{pic-torus}
\end{figure}

We need to prove that $\bar{h}$ does not depend on the choice of $f, F(a), F(b)$. To do this, assume we choose a different $f'$ and therefore different loops $l'_1,l'_2$ in $X$, different homotopy $F'(a)$ and $F'(b)$, and different element $h'\in\pi_2(X)$. To compare $h$ and $h'$ we need to fix a base point. Defined $t$ to be the element in $\pi_2(X)$ determined by $F(a), F'(a)$ and the homotopy between $f, f'$; see \reffig{pic-torus}(b). Also from this figure, we know that
\begin{equation}
    h'=h+t-t^b+s-s^a,
\end{equation}
therefore $\bar{h}=\bar{h'}$.

The inverse map is easy to define. Therefore we have proved \refeq{eq-torusapp}.

\subsection{The action of $\pi_1(X_n)$ on $\pi_2(X_n)$: proof of \refeq{eq-permutation}}\label{app-permutation}

In this section, we prove that the action $t\mapsto t^a$ is determined by \refeq{eq-permutation}, which we rewrite here for convenience:
\begin{equation}\label{eq-permutationapp}
    (t_1,\cdots,t_n)\mapsto (t_{a(1)},\cdots,t_{a(n)}).
\end{equation}

To see this, consider the projection
\begin{equation}
    X_n=(\Conf_n\times F_n)/S_n\xrightarrow{j} F_n/S_n,
\end{equation}
which induces an isomorphism on $\pi_2$ and a surjection $B_n\to S_n$ on $\pi_1$. Therefore, the action of $\pi_1(X_n)$ on $\pi_2(X_n)$ factorizes through the action of $\pi_1(F_n/S_n)=S_n$ on $\pi_2(F_n/S_n)$:
\begin{equation}
\begin{tikzcd}
\pi_2(X_n) \arrow[r,"B_n"] \arrow[d,"\cong"] & \pi_2(X_n) \arrow[d,"\cong"] \\
\pi_2(F_n/S_n) \arrow[r,"S_n"]& \pi_2(F_n/S_n)
\end{tikzcd}.
\end{equation}
Geometrically, the action of $[\gamma]\in\pi_1(X_n)$ on $\pi_2(X_n)$ is given by any homotopy $f_t: S^2\to X_n$ such that $f_t(s_0)=\gamma(t)$ (here $s_0$ is a base point on $S^2$); under projection $j$, $j\circ f_t$ gives a homotopy $S^2\to F_n/S_n$ and therefore an action of $p([\gamma])\in S_n$ on $\pi_2(F_n/S_n)$.

We now show that the action of $\pi_1(F_n/S_n)=S_n$ on $\pi_2(F_n/S_n)$ is given by \refeq{eq-permutationapp}. Indeed, assuming the loop $\gamma$ in $F_n/S_n$ is lifted to $\tilde{\gamma}$ in $F_n$, $\tilde{\gamma}(1)=g\tilde{\gamma}(0)$ where $g\in S_n$. Then a homotopy $S^2\to F_n/S_n$ corresponding to the $[\gamma]$ action will be lifted to a homotopy that deforms the map $\tilde{f_0}:S^2\to F_n$ to $\tilde{f_1}:S^2\to F_n$ such that $\tilde{f_1}(s_0)=\tilde{\gamma}(1), \tilde{f_0}(s_0)=\tilde{\gamma}(0)$. In order to identify the corresponding element of $\tilde{f_1}$ in $\pi_2(F_n/S_n)$, one just needs to consider $g^{-1}\circ \tilde{f_1}$ since they ($g^{-1}\circ \tilde{f_1}$ and $\tilde{f_1}$) are the same map after projection to $F_n/S_n$ and $g^{-1}\tilde{f_1}(s_0)=\tilde{f_0}(s_0)$ is the correct base point; see \reffig{pic-permutation} for illustration of the above argument.
\begin{figure}[t]
    \centering
    \includegraphics[width=\columnwidth]{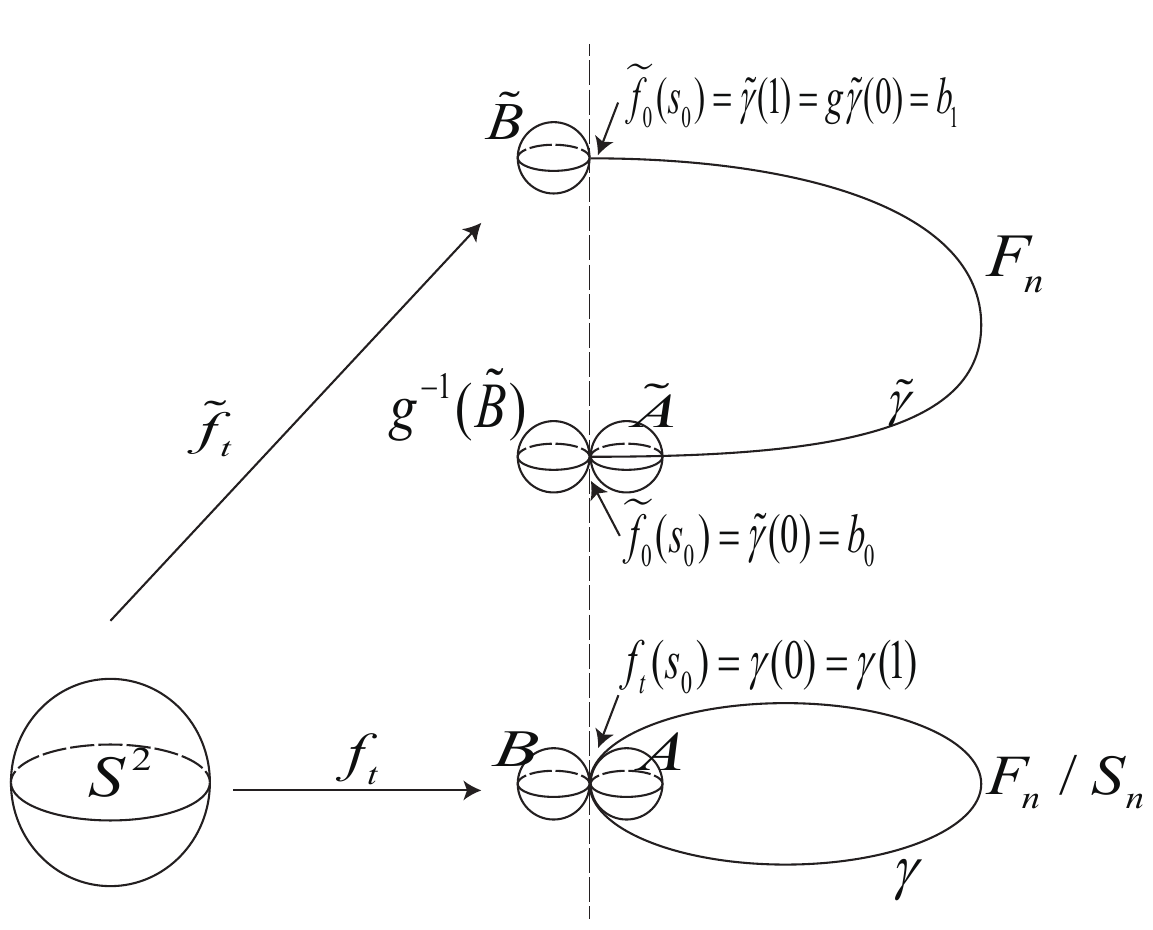}
    \caption{An illustration for the proof of \refeq{eq-permutation}. Assuming $[\gamma]$ action takes an element in $\pi_2(F_n/S_n)$ (represented by $A$) to $B$, then the lift $\tilde{f_t}$ will be a homotopy from $\tilde{A}$ to $\tilde{B}$. To identify the corresponding element of $\tilde{B}$ in $\pi_2(F_n/S_n)$, just consider $g^{-1}(\tilde{B})$ since it is the same as $\tilde{B}$ under projection.}\label{pic-permutation}
\end{figure}

Now we identify $g^{-1}\circ \tilde{f_1}$ in $\pi_2(F_n)=\Z^{n-1}$ according to the injection $\pi_2(F_n)\xrightarrow{\partial}\pi_1(\C^{*n})$ in \refeq{eq-seqm2}. Recall that the boundary map $\partial$ is defined by a homotopy lifting. For example, to identify $\partial(\tilde{f_1})$, one regards $\tilde{f_0}:S^2\to F_n$ as a map $I^2\to F_n$, where $\tilde{f_0}(\partial I^2)=\{b_0\}$; then as a homotopy $H_t: I^1\to F_n$. Then lift $H_0$ along into $GL(n)$. This is just the trivial map to a point, say $e_0$. Then use the relative homotopy lifting property to lift $H_t$ for $t\in I$. $H_1(I)$, which is the lift of $\tilde{f_0}$ on $I\times\{1\}$, is now a loop based on $e_0$, which induces an element in $\pi_1(\C^{*n})$.

\begin{figure}
    \centering
    \includegraphics[width=\columnwidth]{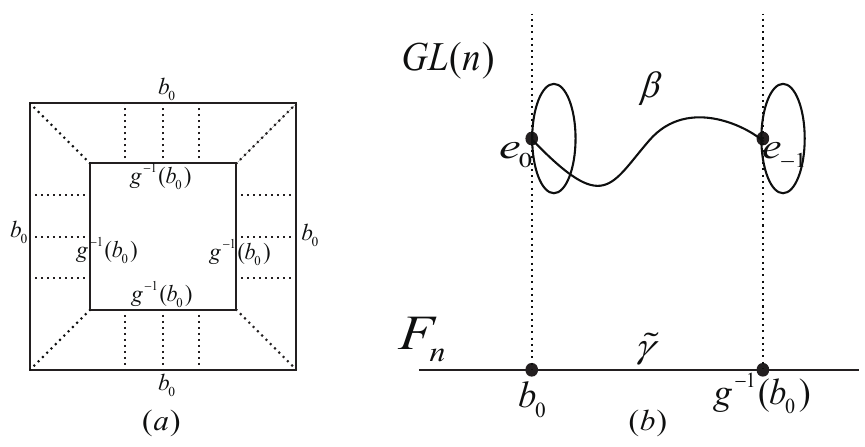}
    \caption{(a) The definition of $g^{-1}\circ \tilde{f_1}$ (corresponds to $g^{-1}(\tilde{B})$ in \reffig{pic-permutation}). Regard $S^2$ as $I^2$ with boundary points identified, then draw a smaller square inside it. Define the map on the inner square as $g^{-1}\circ \tilde{f_0}$ [would be $g^{-1}(\tilde{A})$ in the notation of \reffig{pic-permutation}], so that the inner boundary maps to $g^{-1}(b_0)$. Then one connects the inner and outer boundaries by the paths $g^{-1}(\tilde{\gamma})$. One gets a well-defined map from $I^2$ to $F_n$, with the outer boundary mapping to $b_0$. (b) Illustration of the homotopy lifting process.}\label{pic-homotopy}
\end{figure}

In our case, $g^{-1}\circ \tilde{f_1}$ is just given by \reffig{pic-homotopy}(a). We can construct the lifting explicitly. First note that $S_n$ has an action on $GL(n)$ by column transformation, which is the lift of its action on $F_n$. We lift the path $g^{-1}(\tilde{\gamma})$ in $F_n$ to a path $\beta$ in $GL(n)$ starting at $e_0$. We can make it end at $e_{-1}=g^{-1}e_0$ by gradually changing the phases of each column vector along the path. Now the homotopy lifting is defined as follows [see \reffig{pic-homotopy}(b)]. For $t\in [0,1]$, scan the square in \reffig{pic-homotopy}(a) from bottom to up. For small $t$ (before touching the inner square), just lift the homotopy along $\beta$. Then one lifts the homotopy inside the inner square by $g^{-1}\circ$ the homotopy lifting of $\tilde{f_0}$. After one passes the inner square, one can just move $e_{-1}$ to $e_0$ by shrinking the line $\beta$. The homotopy class ($n$ integers) of the loops on fibers is invariant (For example, since we are only looking at the bundle over an open path $\tilde{\gamma}$, we can regard it as a trivial bundle). The final lifting is a loop on the fiber over $b_0$ with base point $e_0$. It is easy to see this loop corresponds to \eqref{eq-permutation}.

\bibliography{nH_classification}

\end{document}